\documentclass[reqno]{amsart}
\usepackage{amsmath,amssymb}
\usepackage{graphics,epsfig,color}
\usepackage{booktabs,siunitx}  
\usepackage[mathscr]{euscript}
\usepackage{verbatim}
\usepackage{hyperref,color}
\usepackage{wasysym}
\usepackage[dvipsnames]{xcolor}
\usepackage[framemethod=default]{mdframed}
\usepackage[Lenny]{fncychap}

\usepackage{tikz}

\usepackage[all]{xy}
\xyoption{all}

\newtheoremstyle{theorem}{}{}{\itshape}{}{\bfseries}{.}{.5em}{}
\theoremstyle{theorem}
\newtheorem{theorem}{Theorem} 
\newtheorem{proposition}[theorem]{Proposition}

\newtheoremstyle{definition}{}{}{}{}{\bfseries}{.}{.5em}{}
\theoremstyle{definition}
\newtheorem{definition}[theorem]{Definition}
\newtheorem{assumption}[theorem]{Assumption}

\newtheoremstyle{remark}{}{}{}{}{\bfseries}{.}{.5em}{}
\theoremstyle{remark}
\newtheorem{remark}[theorem]{Remark}
\newtheorem{example}[theorem]{Example}

\numberwithin{subsection}{section}
\numberwithin{paragraph}{subsection}

\renewcommand{\div}{\textrm{\,\,div\,}}
\newcommand{\Ker}{\textrm{Ker\,}}
\newcommand{\Ima}{\textrm{Im\,}}
\newcommand{\us}{\underset}
\newcommand{\os}{\overset}
\newcommand{\mc}[1]{{\mathcal #1}}
\newcommand{\mtt}[1]{{\mathtt #1}}
\newcommand{\mf}[1]{{\mathfrak #1}}
\newcommand{\bb}[1]{{\mathbb #1}}

\newcommand{\rt}[1]{\textcolor{red}{#1}}

\newcommand{\vphi}{\varphi}
\newcommand{\p}{\partial}

\newcommand{\e}{\emph}

\newcommand{\veps}{\varepsilon}

\renewcommand{\hat}{\widehat}

\renewcommand{\j}{\jmath}

\definecolor{light}{gray}{.9}

\title[Discrete  calculus with cubic cells on discrete manifolds]{Discrete  calculus with cubic cells on discrete manifolds  }
\author[L.\ De Carlo]{Leonardo De Carlo}
\address{Leonardo De Carlo\hfill\break\indent Departamento de Matemática, Instituto Superior Técnico\hfill\break\indent Av. Rovisco Pais 1049-001 LISBOA Portugal}\email{neoleodeo@gmail.com}

\begin{document}

\begin{abstract}

This work is thought as an operative guide to "discrete exterior calculus" (DEC), but at the same time with a rigorous exposition. We present  a  version of  (DEC) on "cubic" cell, defining it  for discrete manifolds. An example of how it works, it is done on the discrete torus, where usual Gauss and Stokes theorems are recovered.

\bigskip


\smallskip

\end{abstract}

\maketitle
\thispagestyle{empty}




\section{Introduction}

Discrete exterior calculus (DEC) is  motivated by
potential applications in computational methods for field theories (elasticity, fluids, electromagnetism) and in
areas of computer vision/graphics,  for these applications see for example   \cite{DeGDT15}, \cite{DMTS14} or   \cite{BSSZ08}.  DEC is developing as an alternative approach for computational science to the usual discretizing process from the continuous theory. It considers the discrete mesh as the only  thing given and develops an entire calculus using only discrete combinatorial and geometric
operations. The derivations may require that the objects on the discrete mesh, but not the mesh itself, are
interpolated as if they come from a continuous model. Therefore DEC could have interesting applications in fields where there isn't any continuous underlying structure as Graph theory \cite{GP10} or problems that are inherently discrete since  they are defined on a lattice \cite{AO05} and it can stand in its own right as theory that parallels the continuous one.  The language of DEC is founded on the concept of discrete differential form, this characteristic allows to preserve in the discrete context some of the usual geometric and topological structures of continuous models, in particular the Stokes'theorem 
\begin{equation}\label{eq:Stokes}
\int_\mc{M} d\omega=\int_{\partial\mc{M}}\omega.
\end{equation}
 Stating that a differential form $\omega$ over the boundary of some orientable manifold  $\mc{M}$ is equal to the integral of its exterior derivative $d\omega$ over the whole of $\mc{M}$. Equation (\ref{eq:Stokes}) can be considered the milestone to define the discrete exterior calculus since it contains the main objects to set a discrete exterior calculus, namely the concepts of discrete differentials form, boundary operator and  discrete exterior derivative, moreover it is very natural in the discrete setting.
A   qualitative review for DEC   is  \cite{DKT08}, while to have a deeper view we suggest \cite{Cra15} and \cite{Hir03}. 

\vspace{0.5cm}
First we are establishing the necessary objects to introduce the discrete analogous of differential forms, i.e. the discrete exterior derivative $d$ and its adjoint operator $\delta$ in the context of a cubic cellular complex for a  lattice on a discrete manifold. In this way we can  state the Hodge decomposition. This is done in section \ref{sec:DEC}, here the ideas for the construction of DEC  come from \cite{DKT08} and \cite{DHLM05}, but respect to this works we  present more precise definitions for a manifold setting and   cubic cells instead of simplexes, the concepts of comparability, consistency  and local orientation are introduced to define rigorously what is a  discrete manifold. Here we are not exposing  the continuous theory, good references for a treatment of Geometry with differential forms (and its application) to compare with DEC are \cite{AMR88,Fla89} and \cite{Fra12}.  
Then, in section \ref{sec:disop}, we are illustrating how these operators work up to dimension three using as manifold the discrete torus.

\section[DEC on cubic mesh]{Discrete exterior calculus on cubic mesh}\label{sec:DEC}

Intuitively, \e k-differential forms are objects that can be integrated on a \e k dimensional region of the space. For example 1-forms are like $dF=f(x)dx$ or $dG=\frac{\partial G}{\partial x}dx+\frac{\partial G}{\partial y}dy+\frac{\partial G}{\partial z}dz$,  which can be integrated respectively over a interval in $\mathbb{R}$ or  over a curve in $ \mathbb{R}^3 $. With this idea in mind discrete differential forms are going to be defined. As we said we are working  with an abstract cubic complex, instead of a simplicial one. This abstract complex can be tought as a collection of discrete sets of maximal  dimension  $ n $. This collection could be previously derived  from  a continuous structure on a manifold $\mc{M}$ of dimension $n$.

\subsection{Primal cubic complex and dual cell complex}

The next definitions fix in an abstract way  the objects on which DEC operates, the language is the typical one in  algebraic topology \cite{Mun84}.
\begin{definition}\label{d:simdef}
	A \emph{k-simplex} is the convex span $\mf s_k=\{v_0,v_1,\dots,v_k\}$ of $k+1$ geometrically independent points of $ \mathbb{R}^N $ with $ N\geq k $, they are called \emph{vertices of the k-simplex} and  $ k $ its  \emph{dimension}.  A simplex  $ s_k=(v_0,v_1,\dots,v_{k+1}) $ is \e{oriented}   assigning one of the two possible   equivalence classes of   ordering of its vertices $ v_i $. Two orderings are in the same  class if they differ for an even permutation, while they are not for an odd permutation.  The \e k-simplex with same vertices but different ordering from $  s_k $ is said to have \e{opposite orientation} and   denoted with $ -s_k $.  
\end{definition}

One orientation of a simplex can be called conventionally \e{positive} and the opposite one  \e{negative}. 
\begin{definition}[Orientation convention\footnote{This convention tell us the following. We consider $ \mathbb{R}^k $  with a right handed orthonormal basis $ e_1,\dots, e_k $. A simplex $s_1=(v_0,v_1)$ embedded in $ \bb R^1 $  can take orientation from $v_0$ to $v_1$, let's call it positive assuming $ v_1>v_0 $, otherwise from $v_1$ to $v_0$, that is negative. A simplex $s_2=(v_0,v_1,v_2)$ embedded in $ \bb R^2 $  can take anticlockwise orientation with the normal pointing outside the plane along the "right hand rule", let's call it positive,  otherwise clockwise with the normal pointing outside the plane along the "left-hand rule", that is negative. A simplex $s_3=(v_0,v_1,v_2,v_3)$ embedded in $ \bb R^3 $ can take  orientation along the "screw-sense" about the simplex embodied in the familiar "right-hand rule", let's call it positive, otherwise  orientation along the "left-hand rule", that is negative.} for simplexes]\label{d:orsim}
	A way to define the sign of a \e k-simplex $ s_k=(v_0,\dots,v_{k+1}) $ is that of embedding it  in $ \mathbb{R}^k $ equipped with a  right handed orthonormal basis and saying it is positive oriented  if $ \det(v_1-v_0,v_2-v_0,\dots,v_{k+1}-v_0)>0 $ and negative in the opposite case.
\end{definition}

"Inside" a \e k-simplex we can individuate  some proper simplexes, we need to define this and how they relate with the "original" one.
\begin{definition}
	A \emph{ j-face} of a \e k-simplex is any \emph{j}-simplex ($j<k$) spanned by  a proper  subset of  vertices of $ \mf s_k $, this gives a strict  partial order relation $\mf s_j\prec \mf s_k$ and if $ \mf s_j $ is a face of $ \mf s_k $ we denote it $ \mf s_j (\mf s_k)$. A \e j-face is \e{shared} by two \e k-simplexes  $ (j<k) $ if it is a face of both.
\end{definition}

We will need to say    when it is possible and how to compare two \e k-simplexes (of the same dimension), i.e. their reciprocal orientations. The idea of next definition is that this is possible when there exists an hyperplane where  both the \e k-simplexes lie.

\begin{definition}\label{def:consistency}
	We say that two   \e k-simplexes in $ \bb R^N $ with $ N\geq k $ are \e{comparable} if they belong to the same \e k-dimensional hyperplane. Moreover, we say that two comparable oriented simplexes are \e{consistent} if they have same  orientation sign.
\end{definition}

The orientation sign refers to definition \ref{d:orsim}. The condition for two \e k-simplexes to be in the same hyperplane is equivalent   to ask that   any convex span of $ k+2 $ points  chosen from the union of their vertices is not a \e {(k+1)}-simplex. Another relevant  concept in DEC is that one of  induced orientation, that is the orientation that  a face of a simplex inherited when the last one is oriented.  

\begin{definition}\label{def:indor}
	We call \emph{induced orientation} by an oriented \emph k-simplex $ s_k $ on a \e j-face $ s_j(s_k) $  the corresponding  ordering of its vertices in the sequence ordering the vertices   of $ s_k $. If  two \e k-simplexes $ s_k $ and $ s'_k $ induce opposite orientations on a shared \e j-face we say that the \e j-face \e{cancels}.
\end{definition}

Now we introduce the definition of \e k-cube.

\begin{definition}\label{def:s-dec}
	A \e k-cube $\mf c_k=\{v_0,v_1,\dots,v_{2^k-1}\}$  is the convex span of $ 2^k $ points of $ \bb R^N $ with $ N\geq k $ such that there exist $ k! $ different \e k-simplexes having their $ k+1 $ vertices chosen between the \e{vertices} $ v_i $ of $ \mf c_k $ and sharing two by two only one \e {(k-1)}-face. Moreover, vertexes are extremal\footnote{Extremal means that a vertex  can not be written as convex combination of the other vertexes.} points of the convex combination. Each one of these $ k! $ simplexes $ \mf s^i_k $ is said \e{a proper  k-simplex of $ \mf c_k $} and we denote it $ \mf s^i_k(\mf c_k) $ where $ i\in\{0,1,\dots,k!\} $. The dimension of $ \mf c_k $ is $ k $. Note that there is more than one  way to choose these $ k! $ proper \e k-simplex and we call each of them a \e{simplicial decomposition of $ \mf c_k $} denoted $  \varDelta \mf c_k= \os{k!}{\us{i=1}{\cup}}\mf s^i_k $.
\end{definition}
When we don't need to specify the index $ i $ of the these internal simplexes we  omit it. 
\begin{remark}
	A shared \e{(k-1)}-face $ \mf s_{k-1}(\mf s^i_k(\mf c_k))=\mf s_{k-1}(\mf s^{i'}_k(\mf c_k)) $( with $ i\neq i' $) is inside $ \mf c_k $ and not on its boundary. A precise definition of boundary for simplexes and cubes will be given later.
\end{remark}

\begin{definition}
	A \emph{ j-face} of a \e k-cube is any \emph{j}-cube ($j<k$) spanned by  a proper  subset of  vertices of $ \mf c_k $ and not intersecting its interior, this gives a strict  partial order relation $\mf c_j\prec \mf c_k$ and if $ \mf c_j $ is a face of $ \mf c_k $ we denote it $ \mf c_j (\mf c_k)$. A \e j-face is \e{shared} by two \e k-cubes  $ (j<k) $ if it is face of both.
\end{definition}

The concept of orientation for \e k-cubes follows  from that one for \e k-simplexes.

\begin{definition}\label{d:cubor}
	A \e k-cube $ c_k=(v_0,v_1,\dots, v_{2^k-1}) $ is \e{oriented} assigning to each  \e k-simplex  $ \mf s_k^i $ in a simplicial decomposition of $ \mf c_k $ an orientation such that  the  \e{(k-1)}-faces $ s_{k-1}(s^i_k(c_k)) $ that they share cancel\footnote{Namely,  on them, it is induced an opposite orientation,  see definition \ref{def:indor}.}. Two  \e{oriented simplicial decompositions} are in the same equivalence class of orientation if any two comparable  not shared  \e {(k-1)}-simplexes $ s_{k-1}(s^i_k(c_k)) $ (i.e. lying on the same \e{(k-1)}-face on the boundary of $ \mf c_k $) are consistent. We denote $ \varDelta  c_k=\os{k!}{\us{i=1}{\cup}}s^i_k $ an oriented simplicial decomposition.
\end{definition}

\begin{definition}
	Analogously to definition \ref{def:consistency} for simplexes, two \e k-cubes are \e{comparable} if they lie in the same \e k-dimensional hyperplane, while we say that they are \e{consistent} if the \e k-simplexes of the two simplicial decompositions are consistent.
\end{definition}
It is enough to check the consistency between  any \e k-simplex in the simplicial decomposition of one of the two \e k-cube and any \e k-simplex in the decomposition of the other \e k-cube because of definition \ref{d:cubor}.

\begin{example}
	Consider the  2-cube $\mf c_2=\{v_0,v_1,v_2,v_3\} $. A simplicial decomposition is given by  $ s^A_2=(v_0,v_1,v_3) $ and $ s^B_2=(v_1,v_2,v_3) $. Indeed let $ \mf s_1=\{v_1,v_3\} $ be the 1-simplex shared by $ s^A_2 $ and $ s^B_2$, then they cancel on $ \mf s_1 $ because $ s_1(s^A_2)=(v_1,v_3)$ and $s_1(s^B_2)=(v_3,v_1)=-(v_3,v_1) $. Another decomposition is that one given by $  s^C_2=(v_0,v_2,v_3) $ and $  s^D_2=(v_0,v_1,v_2) $. These two decompositions  are in the same equivalence class  because they induce on the 1-simplexes  $ \{v_0,v_1\} $, $\{v_1,v_2\} $, $ \{v_2,v_3 \} $ and $\{v_3,v_0\} $ consistent orientations.  
\end{example}

\begin{proposition}\label{p:cubor}
	There are only two possible equivalence classes of orientation\footnote{ [Orientation convention for cubes]The result of this simplicial decomposition is that also a \e k-cube has only two possible orientations.
		In one dimensions a 1-cube is also a 1-simplex;
		in two dimensions a 2-cube $(v_0,v_1,v_2,v_3)$ can be anticlockwise oriented (positive) with the normal pointing outside the plane along the "right hand rule" or clockwise oriented (negative) otherwise clockwise with the normal pointing outside the plane along the "left-hand rule";
		in three dimensions a 3-cube $(v_0, \dots,v_7)$ can have, looking at it from outside,  all the faces  anticlockwise oriented (positive) or viceversa all clockwise (negative). These two possibilities corresponds respectively  to have all the  normals to its faces pointing outside or inside the volume.} for a k-cube $ \mf c_k $, when an orientation $ c_k $ is assigned the other one is denoted $ -c_k $. One orientation can be  conventionally defined to be positive and the other negative. 
\end{proposition}

\begin{proof} Once a single simplex $ \mf s_k(\mf c_k) $ is oriented, the proposition follows from the facts that  a simplex can be oriented in only two ways because the cancelling condition on  the shared \e{(k-1)}-faces in a simplicial decomposition force all the others to assume an orientation propagating to the entire cubic cell. 
\end{proof}

Now the concept of induced orientation for simplexes can be transferred to cubes.

\begin{definition}\label{d:indorcub}
	We call \emph{induced orientation} by an oriented \emph k-cube\footnote{We saw in definition \ref{d:cubor} and proposition \ref{p:cubor} that  a \emph k-cube $ \mf c_k $ is oriented trough the orientation of  the \emph k-simplexes of its decomposition, hence we talk equivalently of  orientation  induced by the oriented \emph k-simplexes of the simplicial decomposition of $ c_k $.}  $ c_k $ on a \e j-face   $ c_j(c_k) $ the orientation assigned on it by the inductively oriented  \emph{j}-simplexes $ s_{j}(s_k(c_k)) $  of its simplicial decomposition. If  two \e k-cubes induce opposite orientations on a shared \e j-face we say that the face \e{cancels}.
\end{definition} 

We do an example considering the  oriented 2-cube   $c_2=(v_0,v_1,v_2,v_3)$. Let be  $ s^A_2=(v_0,v_1,v_3) $ and $ s^B_2=(v_1,v_2,v_3) $ the 2-simplexes  orienting a simplicial decomposition of  $ c_2 $. The inductively oriented 1-faces  $ c_1(c_2) $ are $ (v_0,v_1) $,$ (v_1,v_2) $, $ (v_2,v_3) $ and $(v_3,v_0) $.

Our intention is to operate on object made by many cubes, like lattice. So we introduce collections of cubes suitable to define a discrete calculus.  Later we will restrict ourselves to the case of discrete manifold. 

\begin{definition}
	A \emph{cubic complex} $\mc{C}$  of \e{dimension n} is a finite collections of  \emph{elementary} cubes $ \mf c_k $, called also \emph{cells},  such that $ 0 \leq k\leq n $, every face of an elementary cube  is in $\mc{C}$ and the intersection of any two cubes of $\mc{C}$ is either empty or a face of both. To each cubes is  assigned  an orientation $ c_k $.  The \e{local orientation} of $ \mc C $ is the orientation of the cubes of dimension $ n $. We denote with $ |\mc C| $ the topological set of $ \mathbb{R}^N $ ($ n\leq N $) given by the union of all \e n-cubes $ \mf c_n $ in $ \mc C $. 
\end{definition}

Fixing an  orientation of the \emph n-cube is like orienting the  \e n-volume of the \e n-hyperplane containing it or equivalently the \e n-volume of the space  $ \mathbb{R}^n $ where it can be embedded. For a  discrete manifold is like to orient  the tangent space for a continuous one.
The meaning of having finite collections of cubes is that  to deal with compact sets in the continuous case.

In particular we are interested in cubic complexes which are  discrete version of an orientable compact  boundaryless manifold $ \mc M $ of dimension $ n $. The idea to define a discrete manifold of dimension $n$ is that to have a cubic complex (in local sense)   topologically equivalent  to a \e n-ball in $\mathbb{R}^n$. Moreover we want to orient a discrete manifold, this is possible using the "cancelling" notion of definition \ref{d:indorcub}.

\begin{definition}\label{def:dMan}
	A \emph{cubic complex} $\mc{C}$  of dimension \e n is a  \emph{discrete manifold} (boundaryless) if every  \e{(n-1)}-cube is shared exactly by two \e n-cubes. A manifold is \e{orientable} if the orientations of all \e n-cubes can be chosen such that  every shared \e{(n-1)}-face cancels.
\end{definition}

The meaning of this definition is that  if we consider a \e{(n-1)}-cube  ${\mf c}_{n-1}\in\mc{C}$ the set $\us{\mf c_n:\,\mf c_{n-1}\in \,\mf c_n}{\cup}\mf c_n$   is  simply connected and homeomorphic to a unit \e n-dimensional ball.

In DEC  important concepts are the ones of dual cell and dual complex. To define a centre of a cube we use barycentric coordinates,  while in \cite{DHLM05} they use the concept of circumcentre for a simplicial complex. This last concept is very simple in the simplicial case but for our case we should introduce the concept of Voronoi diagram \cite{JosTheo13} and we prefer to avoid this.

\begin{definition}\label{def:centre}
	The \emph{centre} of a \emph k-cube $\mf c_k$ is the barycentre of its vertices, denoted  with $\mathtt{b}(\mf c_k)$ $  $. 
\end{definition}

\begin{remark}
	The union $ c_k\cup c_k'/s_k\cup s_k' $ of two comparable consistent \e k-cubes/\e k-simplexes sharing a \e{(k-1)}-face inherits an orientation that is consistent with that one of $ c_k/s_k $ and $ c_k'/s'_k $, in the sense that any \e k-simplex that  can be generated by the vertices in the union and  properly contained in it is defined to be consistent with  $ c_k/s_k $ and $ c'_k/s'_k $. 
\end{remark}

	\begin{definition}
		We define the operation +  as  $  \mf c_k+\mf c'_k:=\mf c_k\cup \mf c'_k $,   when $ \mf c_k=\mf c'_k $ we set  $ \mf c_k+\mf c_k=2 \mf c_k:={\mf c_k\cup \mf c_k} $ as multiset\footnote{ Multisets are sets that can differentiate for  multiple instances of the same element, for example $ {a,b} $ and $ {a,a,b} $ are different multisets.} $ \{\mf c_k,\mf c_k\} $. For oriented cubes $ c_k $ we define an analogous + operation and in addition we define the inverse element, that is  $ c_k+(-c_k)=\emptyset $. The same operation is defined for \e k-simplexes substituting the occurrences of $ \mf c_k $ and $ c_k $ respectively with $ \mf s_k $ and  $ s_k  $.
\end{definition}

The dual of a \emph k-cube, called \emph{dual cell}, is derived from the  duality operator (see next definition \ref{def:dualop}) $\ast:c_k\rightarrow \ast(c_k)$ and   the set of dual cells of a cubic complex will be the \emph{dual complex} $\ast \mc{C}$. Remember that when two \e k-cubes induce two opposite orientation on a shared  \e{(k-1)}-face this last one cancels.

{\begin{definition}\label{def:dualop}
	For a discrete manifold $ \mc C $ of dimension $ n $  the \emph{ duality map} acts on a \emph k-cube $ \mf c_k $ giving  the   \e{(n-k)}-\e{dual cell} $ \ast(\mf c_k) $ obtained with the following union of  (\e{n-k})-simplexes 
	\begin{equation}\label{eq:*op}
	\ast(\mf c_k)=\us{\mf c_n:\mf c_k\prec \mf c_n}{\bigcup}\,\,\us{ \mf c_k\prec \mf c_{k+1}\prec\dots\prec\ \mf c_n}{\bigcup}\{\mathtt{b}(\mf c_k),\mathtt{b}(\mf c_{k+1}),\dots,\mathtt{b}(\mf c_n)\},
	\end{equation}
	where in both union $ \mf c_k $ is fixed, $ \mf c_n $ varies on the first union and fixed in the second one, while the \e{(n-k-2)}-tuple $ \mf c_{k+1},\dots,\mf c_{n-1} $   vary on the second union according to the rule specified in subscript. 
	For an oriented \e k-cube $ c_k=(v_0,\dots, v_{2^{k-1}}) $ the oriented dual cell $ \ast (c_k) $ is obtained assigning to each \e{(n-k)}-simplex \footnote{These simplexes give a simplicial decomposition for the \e{(n-k)}-cube $ \us{ \mf c_k\prec \mf c_{k+1}\prec\dots\prec\ \mf c_n}{\bigcup}\{\mathtt{b}(\mf c_k),\dots,\mathtt{b}(\mf c_n)\} $.} $\{\mathtt{b}(\mf c_k),\dots,\mathtt{b}(\mf c_n)\}  $  an orientation $(\mathtt{b}(\mf c_k),\dots,\mathtt{b}(\mf c_n)) $ if  the oriented \e n-cube $ (v_0,\dots,v_{k-1},\mathtt{b}(\mf c_k),\dots,\mathtt{b}(\mf c_n)) $ is consistent with the local orientation of $ c_n $ and $- (\mathtt{b}(\mf c_k),\mathtt{b}(\mf c_{k+1}),\dots,\mathtt{b}(\mf c_n))  $ otherwise.  The \emph{dual complex} $*\mc{C}$, or \e{dual discrete manifold} ,is the collection $\{{\ast(c_k)}\}_{c_k\in\mc C} $. 
\end{definition}

\begin{remark}\label{r:asy} 
	The oriented (\e{n-k})-simplexes in the union \eqref{eq:*op} share two by two exactly one \e{(n-k-1)}-face that cancels. 
	In definition \ref{def:dualop} with  the oriented \e n-cubes $ c_n $ and $ -c_n $ we always associate a non-oriented object $ \mtt c(\mf c_n)=c_0  $, while with a 0-cube $ c_0 $ we associate an oriented \e n-cube always consistent with the local orientation. This "asymmetry"  is due to the fact that a vertex  doesn't have an intrinsic orientation. Formally we could fix this considering a vertex $ c_0 $ with two orientations $ \pm c_0 $ such that $*(\pm c_n)=\pm c_0  $ and $ *(\pm c_0)=\pm c_n  $. 
\end{remark}

\subsection{Discrete differential forms and exterior derivative}

At this point we are ready to define  discrete versions of differential forms and Stokes Theorem (\ref{eq:Stokes}). Remind the idea of continuous \emph k-form as  object that can be integrated only on a \e k-submanifold and  defined as a linear map from  \emph k-dimensional sets to $\mathbb{R}$. When  \emph k-dimensional sets are defined on a mesh of a discrete manifold  we call them chains, a linear mapping from  chains to a real numbers is  quite a  natural discrete counterpart of a differential form.

\begin{definition}
	Let $\mc{C}$ be a cubic complex and  $\{c^i_k\}_{i\in I_k}$ the collection of all elementary oriented \e k-cubes in $ \mc C $ indexed by $ I_k $. The space of   \emph{k-chains } $C_k(\mc{C})$ is the space with basis $\{c^i_k\}_{i\in I_k}$ of  the finite formal sums $\gamma_k=\us{i\in I_k}{\sum} \gamma_k^i c_k^i$  where the coefficient $ \gamma_k^i $ is an integer.
\end{definition}

In defining the discrete \emph k-forms we are not technical as usual in algebraic topology, we want just to stress that a discrete \emph k-form is a map from the space of \e k-chains to $\mathbb{R}$.

\begin{definition}
	A discrete \emph{k-form} $\omega^k$ is a linear mapping from $C_k(\mc{C})$ to $\mathbb{R}$, i.e.
	\begin{equation}
	\omega^k(\gamma_k)=\omega^k\left(\us{i\in I_k}{\sum}\gamma_k^i c_k^i\right)=\us{i\in I_k}{\sum}\gamma^i_k\omega^k \left(c_k^i\right) .
	\end{equation}
	We add two forms $\omega_1$ and $\omega_2$ adding their values in $\mathbb{R}$, i.e. $
	(\omega_1+\omega_2)\left(\gamma\right)=\omega_1\left(\gamma\right)+\omega_2\left(\gamma\right)$. The vector space of \emph k-forms is denoted $\Omega^k(\mc{C})$.
\end{definition}
Any discrete \e k-form can be written as finite linear combination respect to a basis $\{\alpha^k_i\}_i$ with   same cardinality of $\{c_k^i\}_i$ and determined by the relaton $\alpha_i^k(c_k^j)=\delta_{ij}$.
So we have a \emph{natural pairing} between chains and discrete forms, that is the bilinear pairing 
\begin{equation}\label{eq:pairing}
[\omega^k,\gamma_k]:=\omega^k(\gamma_k).
\end{equation}
Writing $\omega^k=\underset{i}{\sum}\omega^k_i\alpha^k_i$, here $\omega^k_i$ are real coefficient, the pairing \eqref{eq:pairing} becomes $\underset{i\in I_k}{\sum}\omega^k_ic_k^i$.  So the natural pairing \eqref{eq:pairing} leads to a natural notion of \emph{duality} between chains and discrete forms.
In DEC natural pairing plays the role that integration of forms plays in differential exterior calculus. The two can be related by a discretization procedure, for example in the manifold case, thinking to have a piecewise linear \footnote{In case of non-piecewise linear manifold the discretization process present some technicalities, but it is still possible to give meaning to \eqref{eq:omega^k_d}} manifold that can be subdivided in cubes $\{\sigma_k^i\}_i$ and a differential \emph k-form $\omega^k$, the integration  of $\omega^k$ on each \emph k-cube gives its discrete counterpart $\omega^k_d$, where the subscript \emph d is for discrete, defined as
\begin{equation}\label{eq:omega^k_d}
\omega^k_d(\sigma_k):=\int_{\sigma_k}\omega^k.
\end{equation}
In this way a discrete \emph k-form is  a natural representation of a continuous \emph k-form.
\begin{remark}
	A \emph discrete k-form can be viewed as a \emph{k-field} taking different values on different \emph k-cubes of an oriented cubic complex, e.g. a \emph k-form $\omega^k$ on  a \emph k-cube $c_k$ is such that 
	\begin{equation}\label{eq:omegafield}
	\omega^k(-c_k)=-\omega^k(c_k).
	\end{equation}
	With this in mind, for  a discrete \emph 1-form  we use also the name  \e{discrete vector field}.
\end{remark}

To define a discrete exterior derivative, that will give us a  discrete version of \eqref{eq:Stokes}, we have to introduce a discrete boundary operator. As we did so far the definition for cubes goes trough the one for simplexes.
\begin{definition}
	The \emph{boundary operator} $\partial_k:C_k(\mc{C})\rightarrow C_{k-1}(\mc{C})$ is the linear operator that acts  on an oriented \e k-simplex $s_k=(v_0,\dots,v_{k})$ as
	\begin{equation}
	\partial_k s_k=\partial(v_0,\dots,v_{k})=\us{i=0}{\os{k}\sum}(-1)^i(v_0,\dots,\hat{v}_i,\dots, v_{k}),
	\end{equation}
	where $(v_0,\dots,\hat{v}_i,\dots, v_{k})$ is the oriented \emph{(k-1)}-simplex obtained omitting the vertex $v_i$.  Let $\varDelta c_k$ be a simplical decomposition (see definition \ref{def:s-dec}) of $c_k$, the boundary operator on \e k-cubes acts as   
	\begin{equation}
	\partial_k c_k=\us{s_k\in\varDelta c_k}{\sum}\partial_k s_k.
	\end{equation}
\end{definition}

The boundary of  non-oriented objects is obtained doing the boundary of the correspondent oriented objects  and then considering the resulting sets without orientation.

\begin{example} Consider $c_2=(v_0,v_1,v_2,v_3)$ and $\varDelta c_k=(v_0,v_1,v_3)\cup(v_1,v_2,v_3)$, then $\partial_2 c_2=\partial_2(v_0,v_1,v_3)+\partial_2(v_1,v_2,v_3)=(v_1,v_3)-(v_0,v_3)+(v_0,v_1)+(v_2,v_3)-(v_1,v_3)+(v_1,v_2)=(v_0,v_1)+(v_1,v_2)+(v_2,v_3)+(v_3,v_0)$.
\end{example}
In practice $\partial_k$, applied to $c_k$, gives back the faces of $c_k$ with the orientation induced by $c_k$. 
In other terms this $\partial_k$ extracts the oriented border of  an oriented \emph k-cube. A remarkable property of this operator is that  the border of a border is the void set, therefore
\begin{equation}\label{eq:bofb}
\partial_k\circ\partial_{k+1}=0.
\end{equation}

Now with the duality defined by the  natural paring \eqref{eq:pairing} the time is ripe to introduce the discrete exterior derivative (or coboundary operator) $d^k:\Omega^k(\mc{C})\rightarrow\Omega^{k+1}(\mc{C})$  defined by  duality \eqref{eq:pairing} and  the boundary operator.
\begin{definition}\label{d:dder}
	For  a cubic complex $ \mc{C} $ the \emph{ discrete exterior derivative (or coboundary  operator)}    is the linear operator $d^k:\Omega^k(\mc{C})\rightarrow\Omega^{k+1}(\mc{C})$ such that 
	\begin{equation}  \label{eq:def d}
	[ d^k\omega^k,c_{k+1}]  = [ \omega^k,\partial_{k+1}c_{k+1}] 
	\end{equation}
	where  $\omega^k\in\Omega^k(\mc{C})$ and $c_{k+1}\in C_{k+1}(\mc{C})$. Moreover we set $ d\Omega^{n}(\mc C)=0 $. Definition \eqref{eq:def d}  is equivalent to
	$ d^k(\omega^k):=\omega^k\circ\partial_{k+1} $.
\end{definition}
From definition \ref{d:dder} and  \eqref{eq:bofb} it is straightforward the property
\begin{equation}\label{dd=0}
d^{k+1}\circ d^k=0.
\end{equation}
\begin{assumption}
	Unless otherwise specified, we are omitting if  operators are referred to the primal complex or its dual. We assume  the right one at the right moment. Moreover for discrete manifolds in definition \ref{def:dMan} the  operators of this section don't change in the dual.
\end{assumption}

From  definition \eqref{eq:def d} and natural pairing \eqref{eq:omega^k_d}  we have  a discrete  Stokes Theorem: consider a chain $\gamma_k$  and a discrete form $\omega^k$ then
\begin{equation}\label{eq:distokes}
\int_{\gamma_k} d^k\omega^k \equiv [ d^k\omega^k,\gamma_k] =  [\omega^k,\partial_k \gamma_k] \equiv \int_{\partial_k \gamma_k}\omega^k.
\end{equation}

\subsection{Hodge star and codifferential}
The  counterpart of $d^k$, denoted with $ \delta^{k+1} $, mapping  a \e {(k+1)}-form into a \e{k}-form is the tool still missing to have all what we need from DEC. Given two \e k-forms $\omega^k_1$ and $\omega^k_2$, this operator is defined as the adjoint of $d$ with respect to the scalar product 
\begin{equation}\label{eq:*scalar2}
\langle \omega_1^k, \omega_2^k \rangle= \underset{i\in I_k}{\sum}\omega^k_{1,i} \omega^k_{2,i}
\end{equation}
This scalar product is the discrete version of formula \eqref{eq:smoothsp} in footnote \ref{fo:star} below.
\begin{definition}
	The discrete codifferential operator\footnote{\label{fo:star}In the smooth case, for a manifold $ \mc{M} $ of dimension \e n,   the Hodge star is the map  $\star: \Omega_k(\mc{M}) \rightarrow \Omega_{n-k} (\mc{M})$, defined by its local metric and the local scalar product of \e k-forms $\langle\langle \omega^k_1,\omega^k_2 \rangle\rangle=(\omega^k_1)^{i_1,\dots,i_k}(\omega^k_2)_{i_1,\dots,_k}$, such that 
		\begin{equation*}\label{eq:*smooth}
			\omega^k_1\wedge\star\omega^k_2:=\langle\langle \omega^k_1,\omega^k_2\rangle\rangle vol^n,
		\end{equation*}
		where $vol^n$ is the volume form on the manifold. Denoting with $ d $ the exterior derivative for differential forms, this operator can be computed through its action on the the basis of \e k-forms $\mathbf{dx}^k=dx^{i_1}\wedge\dots \wedge dx^{i_k}$ ($i_1<\dots<i_k$), that gives back the forms $\star \mathbf{dx}^k=C\mathbf{dx}^{n-k}=C dx^{i_{k+1}}\wedge\dots \wedge dx^{i_n}$ ($i_{k+1}<\dots<i_{n}$) with $C$ is  such that $\mathbf{dx}^k\wedge C \mathbf{dx}^{n-k}=\langle\langle \mathbf{dx}^k,\mathbf{dx}^{k} \rangle\rangle vol^n=vol^n$.
		On a smooth manifold $ \mc{M} $, for details see chapter 14 in \cite{Fra12} or section 6.2 in \cite{AMR88}, the adjoint operator  $ \delta $ of $ d $ is defined respect to the scalar product 
		\begin{equation}\label{eq:smoothsp}
		\langle \omega^k_1,\omega^k_2\rangle:=\int_{\mc{M}}\omega^k_1\wedge\star\omega^k_2=\int_{\mc{M}}\langle\langle \omega^k_1,\omega^k_2\rangle\rangle vol^n.
		\end{equation}
		Thinking about,  \eqref{eq:*scalar2} is the discrete version of the most right term of the last formula \eqref{eq:smoothsp} and it can be introduced in a sophisticated way that "emulates" the continuous case of this footnote  \ref{fo:star}. Let's try to sketch  this parallelism.
		Since the duality $\ast$ maps a primal cell into an only one dual cell and vice versa,  the most spontaneous thing to set a \e{discrete Hodge star} $ \star $  from \e k-forms into  \e {(n-k)}-forms is doing it from  $ \Omega^{k}(\mc{C}) $ into $ \Omega^{n-k}(\ast\mc{C}) $, i.e. $\star:\Omega^k(\mc{C})\rightarrow\Omega^{n-k}(\ast\mc{C})$, this can be defined with the relation $ (\omega^k,c_k)=(\star\omega^k,\ast c_k) $. With this definition \eqref{eq:*scalar2} can be written as $ \langle\omega^k_1,\omega^k_2\rangle:=\sum_{i\in I_k}(\omega^k_1,c_k^i) (\star\omega^k_2,\ast c^i_k)= \underset{i\in I_k}{\sum}\omega^k_{1,i} \star\omega^k_{2,i} $ that is the equivalent of the middle term in \eqref{eq:smoothsp}. For details how to define a discrete wedge product and the counterpart of the \e n-volume form see section 12 of \cite{DHLM05} and \cite{DMTS14}} $ \delta^{k}:\Omega^k(\mc C)\to\Omega^{k-1}(\mc C) $ is defined by $ \delta^0\Omega^0(\mc C)=0 $ and the equation 
	\begin{equation}\label{eq:codef}
	\langle d^k\omega_1^{k}, \omega_2^{k+1} \rangle= \langle \omega_1^{k}, \delta^{k+1}\omega_2^{k+1} \rangle,
	\end{equation}
	where $ \omega^k_1\in\Omega^k(\mc C) $ and $ \omega_2^{k+1}\in \Omega^{k+1}(\mc C) $.
\end{definition}                                    
Also for this operator we have a property analogous to \eqref{dd=0}, that is 
\begin{equation}\label{eq:deltaodelta}
\delta^k\circ\delta^{k+1}=0.
\end{equation}                             
For a   matrix description of the operators of this chapter see \cite{DKT08}. The operators $ \ast,\partial,d,\star $ and $ \delta $ we introduced on $ \mc C $ can be defined also on the dual complex $ \ast\mc C $. In particular when a discrete manifold is considered and the inverse map of duality map $ \ast $ can be written as  
\begin{equation}\label{e:ast}
\ast^{-1}=(-1)^{k(n-k)}\ast
\end{equation}  where $ \ast $ acts on $ \ast \mc C $ as on $ \mc C $ , i.e. when $ \ast\mc C $ is still  made of \e k-cubes, all definitions applies in the same way. In this case\footnote{\label{f:delta}In this case we have also the formula $ \delta^{k+1}\omega^{k+1}=(-1)^{nk+1}\star d^{n-(k+1)}\star\omega^{k+1} $, which is analogous to the one for $ \delta $ in smooth boundaryless maniflods. This is proved using $ \star^{-1}=(-1)^{k(n-k)}\star $, which follows from \eqref{e:ast} and the definition of $ \star  $ in footnote \ref{fo:star}, see also subsection 5.5 in \cite{DKT08}.} the same notation for the discrete operators is used both on $ \mc C $ and $ \ast \mc C $.

\subsection{Hodge decomposition}
 The original work of the Hodge decomposition theorem for finite dimensional complexes is \cite{Eck44}. The Hodge decomposition in our context is as follows.

\begin{theorem}\label{th:HT}
	Let $ \mc{C} $ be a discrete manifold of dimension n and let $ \Omega^k(\mc{C}) $ be the space of   k-forms on $\mc{C}$. The following orthogonal decomposition holds for all k:
	\begin{equation}\label{eq:dhodd}
	\Omega^k(\mc{C})=d^{k-1}\Omega^{k-1}(\mc{C})\oplus\delta^{k+1}\Omega^{k+1}(\mc{C})\oplus\Omega^k_H(\mc{C}),
	\end{equation}
	where $\oplus$ means direct sum and  $ \Omega^k_H(\mc{C})= \{\omega^k|d^{k}\omega^k=\delta^k\omega^k=0\}$ is the space of harmonic forms.
\end{theorem}

\begin{proof}
	Consider  the discrete forms  $ \omega^{k-1}_1\in\Omega^{k-1}(\mc C) $, $ \omega^{k+1}_2\in\Omega^{k+1}(\mc C) $ and $ h^k\in \Omega_H^k(\mc C) $ chosen arbitrary form their respective spaces.
	Since $ \delta^{k+1}  $ is the adjoint of $ d^k $ using  \eqref{dd=0} and \eqref{eq:codef} we have $ \langle d^{k-1}\omega_1,\delta^{k+1} \omega_2 \rangle=\langle d^kd^{k-1}\omega_1, \omega_2 \rangle =0$.  By the definition of $ \Omega^k_H(\mc C) $ and \eqref{eq:codef} we have $ \langle h^k,d^{k-1} \omega^{k-1}_1 \rangle=\langle \delta^k h^k, \omega^{k-1}_1 \rangle=0$, likewise $\langle h^k, \delta^{k+1}\omega^{k+1}_2 \rangle=\langle  d^kh^k, \omega^{k+1}_2 \rangle=0 $. Therefore the spaces $ d^{k-1}\Omega^{k-1}(\mc C)$, $ \delta^{k+1}\Omega^{k+1}(\mc C)$ and $ \Omega^k_H(\mc C) $ are each other orthogonal. A general \e k-form $ \omega^k\in \Omega^k(\mc C)$  belongs to $ \left(d\Omega^{k-1}(\mc{C})\oplus\delta\Omega^{k+1}(\mc{C})\right)^\perp$ if and only if  $ \langle \omega^k,d^{k-1} \omega^{k-1}_1+ \delta^{k+1}\omega^{k+1}_2 \rangle=\langle \delta^k \omega^k, \omega^{k-1}_1 \rangle+\langle  d^k \omega^k, \omega^{k+1}_2 \rangle=0 $ for all $ d^{k-1} \omega^{k-1}_1+ \delta^{k+1}\omega^{k+1}_2\in d\Omega^{k-1}(\mc{C})\oplus\delta\Omega^{k+1}(\mc{C}) $, namely $ d^k\omega^k=\delta^k\omega=0 $. Then  we  showed  $ \left(d\Omega^{k-1}(\mc{C})\oplus\delta\Omega^{k+1}(\mc{C})\right)^\perp=\Omega^k_H$, so the decomposition is complete and generates all $ \Omega^k(\mc C) $.
\end{proof}

\section[Discrete operators  on   $ \mathbb{T}^3 $]{Discrete operators  on  the discrete torus}\label{sec:disop}

In this section we show how the discrete exterior derivative and its adjoint work on a cubic complex of dimension three, namely we work with \e 0-forms, \e 1-forms, \e 2-forms and \e 3-forms, setting discrete equivalent of gradient, curl and divergence operators.
\emph{ We consider a discrete mesh of edge 1 on the discrete torus $ \mathbb{T}^3_N=\mathbb{Z}^3/N\mathbb{Z}^3 $ of side N and we refer explicitly to other dimensions whenever appropiate}.  
Thinking about the bases $dx_i$, $dx_idx_j $ and $ dx_1dx_2dx_3 $ respectively for \e 1-forms, \e 2-forms and \e 3-forms, the parallel of the left   table \ref{tab:forme} between  the smooth case and the discrete case is really natural. For smooth form the action of $ d $ and $ \delta $ can be summarized  as in the right table  \ref{tab:forme}.

\begin{table}
	\centering
	\caption[]{Left table: conceiving continuous  and discrete forms. Right table: $ d $ and $ \delta $ applied to smooth forms.}
	\label{tab:forme}
	\begin{tabular}{|cc|cc|}
		\hline
		\multicolumn{2}{|c|}{smoooth case}&\multicolumn{2}{c|}{discrete case}\\ 
		\hline
		$ \omega^0 $: & scalar field &  $ \omega^0 $: & vertex field \\
		$ \omega^1 $: & vector field &  $ \omega^1 $: & edge field\\
		$ \omega^2 $: & vector  field & $ \omega^2 $: & face field\\
		$ \omega^3 $: & scalar field &  $ \omega^3 $: & cell field\\
		\hline
	\end{tabular}
	\hspace{0.08mm}
	\begin{tabular}{|c|c|c|}
		\hline
		\multicolumn{3}{|c|}{smoooth case}\\
		\hline
		form & $ d $ &  $ \delta $ \\
		\hline
		$ \omega^0 $ &  grad$ \,\omega^0 $ &  0 \\
		$ \omega^1 $ &  curl$ \,\omega^1 $ &  -div$ \,\omega^1 $ \\
		$ \omega^2 $ &  div$ \,\omega^2 $  &  curl$ \,\omega^2 $ \\
		$ \omega^3 $ &   0 & -grad$ \,\omega^3 $ \\
		\hline
	\end{tabular}{ }
\end{table}

Now we are indicating also the discrete operators of last section \ref{sec:DEC} without the index $ k $,  unless otherwise specified, it will be implicit to use the right one according to the \e k-form on which they act. We want to show how to compute $ d $ and $ \delta $. Calculations will be performed in some cases, the others follow mutatis mutandis. We recover also the usual divergence, Gauss and Stokes theorems in our discrete setting. A part proposition \ref{prop:Omega^1_H}, all what  we do in this section for the discrete torus can be done with some extra work and notation for general discrete manifold with non regular mesh (i.e. that can not be defined using the canonical basis $ \{e_1,e_2,e_3\} $).

\subsection{Notation}\label{ss:not}Let  $\{{e_1,e_2,e_3}\} $ be a canonical right handed orthonormal basis.  We define the sets of all the \e{vertices}  
\begin{equation*}\label{key}
	V_N:=\{ x=(x_1,x_2,x_3):x_i=0,e_i,\dots,(N-1)e_i\}.
\end{equation*}
all the oriented \e{edges} $ E_N=\left(\os{3}{\us{i=1}{\bigcup}}E_N^{i,+}\right)\cup \left(\os{3}{\us{i=1}{\bigcup}}E_N^{i,-}\right) $ where 
\begin{equation*}\label{key}
	E_N^{i,\pm}=\{\text{positive/negative oriented \e 1-cubes } e=\pm(x,x+e_i),i\in{1,2,3}\},
\end{equation*} all the oriented \e{faces} $ F_N=\left(\os{3}{\us{i=1}{\bigcup}}F_N^{i,+}\right)\cup \left(\os{3}{\us{i=1}{\bigcup}}F_N^{i,-}\right) $ where 
\begin{equation*}\label{key}
	F_N^{1,\pm}=\{\text{positive/negative oriented  \e 2-cubes } f_1=\pm(x,x+e_2,x+e_2+e_3,x+e_3) \},
\end{equation*} 
\begin{equation*}\label{key}
	F_N^{2,\pm}=\left\{\text{positive/negative oriented \e 2-cubes } f_2=\pm(x,x+e_3,x+e_3+e_1,x+e_1) \right\},
\end{equation*} 
\begin{equation*}\label{key}
	F_N^{3,\pm}=\left\{\text{positive/negative oriented \e 2-cubes } f_3=\pm(x,x+e_1,x+e_1+e_2,x+e_2)\right\}
\end{equation*}
and finally all the oriented \e{cells} $C_N=\left(C_N^{+}\right)\cup \left(C_N^{-}\right)$  where 
\begin{equation*}\label{key}
	C^{\pm}_N:=\left\{\textrm{positive/negative oriented \e 3-cubes } \pm c:\texttt{b}(c)=x+\frac{e_1}{2}+\frac{e_2}{2}+\frac{e_3}{2}\right\}.
\end{equation*}

\begin{remark}\label{r:face}
	Observe that for a face the orientation can be defined as for a \e 2-simplex, see definition \ref{d:simdef} and \ref{d:orsim}, i.e. with the ordering given to the vertices in its sequence.
\end{remark}

We indicate a general  oriented edge,  oriented face and oriented cell respectively with $ e=(x,y) $, $ f $ and $ c $, where $ y=x\pm e_i $ for some $ i $, while for the   non-oriented  ones we use a  calligraphic writing, that is $ \mf e=\{x,y\} $,  $ \mf f $ and $ \mf c $. We indicate in calligraphic also the non-oriented collections just defined above, that is   $ \mc E_N $, $ \mc F_N $ and $ \mc C_N $. 
The collections of \e k-cube defining the discrete Manifolds $ \mc C $  on $ \bb T^3_N $, $ \bb T^2_N $ and $ \bb T_N $ are respectively $ \left\{V_N,E^+_N,F^{+}_N,C_N^+\right\} $, $ \left\{V_N,E^{1,+}_N,E^{2,+}_N,F^{+}_N\right\} $ where $ F^{+}_N $ is defined as $ F^{3,+}_N $ and   $ \left\{V_N,E^{+}_N\right\} $ where $ E_N^+ $ is defined as $ E_N^{+,1} $. 

A general vertex, edge, face and cell field is going to be denoted respectively with $h(x)$, $j(x,y) $,  $\psi(f) $ and  $ \rho(c)$. 

The dual complex $ \ast\mc C $ is constructed on a mesh with the same geometrical structure of  the original one and  obtained translating each vertex of $ (e_1/2,e_2/2,e_3/2) $. We denote with $ \bb T^{n,\ast}_N $ the dual torus  obtained translating the vertexes. So we use an analogous notation and indicate with an index $ * $ the collections and the elements of the dual complex, that is respectively $ V_N^* $, $ E_N^* $, $ F_N^* $ and $ C_N^* $  and $ x^* $, $ e^*=(x^*,y^*) $, $ f^* $ and $ c^* $. 

\begin{remark}\label{r:orob}
	With a vertex $ x\in V_N $ are associated the three edges  $ \{x,x+e_1\}$,$\{x,x+e_2\} $ and $ \{x,x+e_3\} $, consequently we associate with $ x $  also the oriented edges $ \pm (x,x+e_1)$,$\pm(x,x+e_2) $ and $ \pm(x,x+e_3)$. While in two dimension with $ x\in V_N $ are associated $ \{x,x+e_1\}$ and $\{x,x+e_2\} $ and in one dimension only $ \{x,x+e_1\}$.
	
	With a vertex $ x\in V_N $ are associated the three  faces $ \mf f_1,\mf f_2 $ and $ \mf f_3 $ such that    $\mathtt b(\mf f_k)=x+\frac{e_i}{2}+\frac{e_j}{2}  $  where $ i,j,k\in\{1,2,3\} $ and  $ i\neq j \neq k $, consequently we associate with $ x $  also the oriented faces $ \pm f_1$,$\pm f_2 $ and $ \pm f_3$. In two dimensions we don't have a subscript on the faces and only one  face $ \mf f $, defined as $ \mf f_3 $, is associated with  $ x $. 
	
	Finally with a vertex $ x\in V_N $ is associated only one  cell $ \mf c $ such that $ \texttt{b}(\mf c)=x+\frac{e_1}{2}+\frac{e_2}{2}+\frac{e_3}{2} $, consequently the oriented cells $ \pm c $ are associated with $ x $.
\end{remark}

\subsection{\label{ss:d 0to1}$ \mathbf{d:\Omega^0\rightarrow\Omega^1} $} 
Consider $ (x,y)\in E_N $ and compute
$
dh(x,y)=h\circ\partial(x,y)=h(y-x)=h(y)-h(x).
$
Choosing $ (x,y)=(x,x+e_i) $  we define the discrete gradient
\begin{equation}\label{eq:dgrad}
dh(x,x+e_i)=h(x+e_i)-h(x)=:\nabla_i h(x).
\end{equation}
For a \e 1-chain $ \gamma=\os{m}{\us{i=1}{\sum}}(x_k,y_k) $, where $ y_k=x_{k+1} $, setting $ \int_\gamma \nabla h\cdot dl:=\us{k=1}{\os{m}{\sum}}\nabla_{i_k} h(x_{k}) $ we get
\begin{equation}\label{eq:dlineint}
dh(\gamma)=\int_{\gamma}\nabla h\cdot dl=h\circ\partial(\gamma)=\int_{\partial\gamma}h = h(y_m)-h(x_1),
\end{equation}
where $ \p\gamma $ is the boundary of $ \gamma $. This gives us  the discrete versions of  line integral. For completeness we translates in our language the well known result relating gradient fields to  zero integrations on closed paths.
\begin{proposition}\label{prop:prop1form}
	A 1-form $ j(x,y)\in d\Omega^0$ if and only if $ j(\gamma)=\oint_\gamma j=0 $ for all closed path(chain) $ \gamma $ on $ V_N $. Moreover, the vertex function  \begin{equation}\label{eq:h^x(y)}
	h^x(y):=\us{(w,z)\in\gamma_{x\rightarrow y}}{\sum}j(w,z) 
	\end{equation} is such that $ j(y,y')=h^x(y')-h^x(y) $ for every $ (y,y')\in E_N $ and it doesn't depend on the particular path  $ \gamma_{x\rightarrow y} $   from $ x $ to $ y $. Any two functions $ h^{x'}(\cdot) $  and  $ h^x(\cdot) $  differ for  an additive constant.
\end{proposition}

\begin{proof}
	If $ j(x,y)=dh(x,y) $ from \eqref{eq:dlineint} we get $ \oint_\gamma j =0 $ since $ y_m=x_1 $. Now let's prove the opposite implication. Let $ \gamma $ be a closed path and $ x,y $ any two points on it. Since $ \gamma=\gamma_{x \rightarrow y}+\gamma_{y\rightarrow x} $, from $ \us{(w,z)\in\gamma}{\sum}j(w,z)=\us{(w,z)\in\gamma_{x\rightarrow y}}{\sum}j(w,z)+\us{(w,z)\in\gamma_{y\rightarrow x}}{\sum}j(w,z)=0 $, we have $ \us{(w,z)\in\gamma_{x\rightarrow y}}{\sum}j(w,z)=\us{(w,z)\in\gamma'_{x\rightarrow y}}{\sum}j(w,z) $, where $ \gamma'_{x\rightarrow y}=-\gamma_{y\rightarrow x} $. Hence the function \eqref{eq:h^x(y)} doesn't depend on the particular path from $ x  $ to $ y $ and for any $ (y,y')\in E_N $ its gradient is $ h^x(y')-h^x(y)=\us{(w,z)\in\gamma_{x\rightarrow y}\cup(y,y')}{\sum}j(w,z)-\us{(w,z)\in\gamma_{x\rightarrow y}}{\sum}j(w,z)=j(y,y')$. Taking a closed path $ \gamma=\gamma_{x'\rightarrow y}+\gamma_{y\rightarrow x}+\gamma_{x\rightarrow x'} $ we obtain $ h^{x'}(y)=h^x(y)+\us{(w,z)\in\gamma_{x'\rightarrow x}}{\sum}j(w,z) $  because of $ \oint_\gamma j=0 $.
\end{proof}

\subsection{$ \mathbf{d:\Omega^1\rightarrow\Omega^2} $} 
Let   $ f_k\in F^{k,+}_N$ be the face with index $ k $ that is  associated with  $ x $ as in previous notation in subsection \ref{ss:not}, see remark \ref{r:orob}. We have
$ dj(f_k)=j\circ\partial (f_k)= \us{(x,y)\in \p f_k}{\sum}j(x,y)$, that we write
\begin{equation}\label{eq:dj}
dj(f_k)=\us{(x,y)\in \p f_k}{\sum}j(x,y)=:\int_{\p f_k}j\cdot dl.
\end{equation}
We compute $ dj(f_3) $ as
$
dj(f_3)=j(x,x+e_1)+j(x+e_1,x+e_1+e_2)-j(x+e_2,x+e_1+e_2)-j(x,x+e_2),
$
calling $ \nabla_l j_m(x):= j(x+e_l,x+e_l+e_m)-j(x,x+e_m) $, where $ j_m(x):=j(x,x+e_m) $, the last equation becomes
$ dj(f_3)=\nabla_1 j_2(x)-\nabla_2 j_1(x) $, so  $ dj(f_k) $ is also the curl defined as 
\begin{equation*}\label{eq:dcurl}
	dj(f_k)=\varepsilon^{klm}\nabla_lj_m(x)=:(\mathrm{curl\,}j)(f_k),
\end{equation*}
where $ \veps^{klm} $ is the Levi-Civita symbol summed on repeated indexes with the Einstein convention.
Let $ S $ be the \e 2-chain $ S=\os{m}{\us{i=1}{\sum}} f_i $,  where $ \{f_i\}_{i=1}^{m} $ can be any collection  of  oriented faces in $ F_N $. Setting $ \int_S \textrm{curl}j\cdot d\varSigma:=\us{i=1}{\os{m}{\sum}}(\textrm{curl\, }j)(f_i) $,  we have
\begin{equation}\label{eq:dstokes3d}
\int_S\textrm{curl\,}j\cdot d\varSigma=\int_\Gamma j\cdot dl.
\end{equation}
where $  \Gamma=\os{m}{\us{i=1}{\sum}}\p f_i  $. If all the faces in the collection have same orientation, i.e. $ f_i\in F^\pm_N $ for all $ i\in 1,\dots,m $, the chain $  \Gamma=\os{m}{\us{i=1}{\sum}}\p f_i  $ is the closed path  corresponding to  the inductively oriented boundary of the oriented surface $ S $, indeed   when a edge is  shared by two faces $ f_i $ and $ f_j $ it cancels and  gives two contributions with same modulo but opposite sign in $ \int_{\p f_i}j\cdot dl $ and $ \int_{\p f_j}j\cdot dl $. We obtained in this last case  a  discrete version of the usual Stokes theorem.
\subsubsection*{\quad Two dimensions.}In two dimensions  $ dj(f)=\nabla_1 j_2(x)-\nabla_2 j_1(x) $,  where $ f $ is the face associated with $ x $ as $ f_3 $ in notation and   $ \textrm{curl}j $ is a \e 2-form defined  on the faces  lying on the plane of  $ e_1 $ and $ e_2 $. Proceeding like we did in three dimensions we have the discrete Green-Gauss formula.

\subsection{$ \mathbf{d:\Omega^2\rightarrow\Omega^3} $} Consider the $ c\in C_N^+ $,  we compute and define the divergence $ d $ as follows 
\begin{equation}\label{eq:dpsi}
d\psi(c)=\psi\circ\partial(c)=\us{f\in F_N:f\in \p c}{\sum}\psi(f)=:\div\psi(c).
\end{equation}
As in notation the cell $ c $ is associated with the vertex $ x\in V_N $.  Let $ f_i $ with $ i\in \{1,2,3\} $ the three positive faces  associated with $ x $. We have then
$\us{f:f\in \p c}{\sum}\psi(f)=\us{i=1}{\os{3}{\sum}}\psi(f_i+e_i)-\psi(f_i) $. We define the flow across $ c $ as
\begin{equation}\label{eq:ddiv}
\int_{\partial c}\psi\cdot d\varSigma:=\us{i=1}{\os{3}{\sum}}\Phi(f_i),
\end{equation}
where $ \Phi(f_i):=\psi(f_i+e_i)-\psi(f_i) $. Let $ V $ be the 3-chain $ V=\os{m}{\us{i=1}{\sum}}c_i $, where $ \{c_i\}_{i=1}^{m} $ can be any collection  of  oriented cubes in $ C_N $. Setting $ \int_V \textrm{div\,}\psi \,dx:=\us{i=1}{\os{m}{\sum}}(\textrm{div}\psi)(c_i) $, we have
\begin{equation}\label{eq:ddivth3d}
\int_V\textrm{div\,}\psi\, dx=\int_{S} \psi\cdot d\varSigma,
\end{equation}
where $ S=\os{m}{\us{i=1}{\sum}}\p c_i $. If all the cubes in the collection have same sign orientation, i.e. $ c_i\in C_N^\pm $ for all $ i\in 1,\dots,m $, the chain  $ S=\os{m}{\us{i=1}{\sum}}\p c_i $ is the closed surface corresponding to the  inductively oriented boundary of the oriented volume $ V $, indeed  when a face is  shared by two cells $ c_i $ and $ c_j $ it cancels and  gives two contributions with same modulo but opposite sign in $ \int_{\partial c_i}\psi\cdot d\varSigma $ and $ \int_{\partial c_j}\psi\cdot d\varSigma $. 
We obtained in this last case a discrete version of the usual   divergence theorem.

\bigskip
Now we compute the codifferential operator $ \delta $, this can be done using the fact that it is the adjoint of $ d $ with respect to the scalar product \eqref{eq:*scalar2} or using the formula for $ \delta $ in footnote \ref{f:delta} since \eqref{e:ast} is true. We don't write the computation but we give just the explicit form for $ \delta $.

\subsection{$ \mathbf{\delta:\Omega^3\rightarrow\Omega^2} $}
Let   $ f_i\in F^{i,+}_N$ be the face with index $ i $ that is  associated with  $ x $ as in  notation of subsection \ref{ss:not}, see remark \ref{r:orob}, we have
\begin{equation}\label{eq:co3-2}
\delta\rho(f_i)=\rho(c-e_i)-\rho(c)=\us{c':f_i\in \p c'}{\sum}\rho(c')
\end{equation}
With the help of the dual complex we can recover line integrals analogously to subsection \ref{ss:d 0to1}. To do it we should interpret $ \rho $ as 0-form on the dual vertex $ x^*=\ast c $ of a cell $ c $  (see the operator $ \star $ defined in  footnote  \ref{f:delta})   and $ \delta\rho $ as discrete vector field of gradient type ($ \nabla_i\rho(c):=\rho(c)-\rho(c-e_i) $) on the dual edges $ (x^*,y^*)=*(f_i) $ of the faces  $ f_i\in F^+_N $.  The collection of cells  $ c $ of $ C_N $ that would define the dual chain $ \gamma^* $ equivalent to  $ \gamma $ in \eqref{eq:dlineint} should be such that   any two consecutive  cells $ c_i$,   $c_{i+1} $ in the collection  share a common face that cancels.

\subsection{$ \mathbf{\delta:\Omega^2\rightarrow\Omega^1} $}\label{ss:delta3to2}
On an edge $ (x,y)\in 	E_N $ the codifferential operator is 
\begin{equation}\label{eq:co2-1}
\delta\psi(x,y)=\sum_{f:(x,y)\in \p f}\psi(f).
\end{equation}
Let $\{f_j\}_{j=1}^3 $ be the faces associated with $ x $ and let's set $ \nabla_i\psi(f_j):=\psi(f_j)-\psi(f_j-e_i) $, the  codifferential  can be rewritten $ \delta\psi(x,x+e_k)=\epsilon^{klm}\nabla_l\psi(f_m)=: \textrm{curl}\,\psi(x,x+e_k)$ getting the relation
$
\textrm{curl}\,\psi(x,x+e_k)=\sum_{f:(x,x+e_k)\in f}\psi(f)
$.
The definition of curl as integral on the border of a surface as \eqref{eq:dj} and a the discrete version of the Stokes theorem like \eqref{eq:dstokes3d} can be recovered with the dual complex. We should interpret (see the operator $ \star $ defined in  footnote  \ref{f:delta})   $ \mathrm{curl}\,\psi(x,x+e_k) $ as a 2-form on the dual  face $f^*= *(x,x+e_k) $ and $ \psi(f)  $ as discrete vector fields on the dual edges $ (x^*,y^*)=*f $. With a suitable collection of edges $ \{e_i\}_{i=1}^m $ we can define any desired surface $ S^* $.

\subsubsection*{\quad Two dimensions.} 
Formula \eqref{eq:co2-1} is still true, calling $ f $ the face associated to $ x $,  the codifferential is interpreted as an  orthogonal gradient
\begin{equation}\label{eq:dortgra}
\left(\begin{array}{ll}
\delta\psi(x,x+e_1)\\
\delta\psi(x,x+e_2)
\end{array}\right)
=\left(\begin{array}{ll}
\;\;\;\nabla_2\psi(f)\\
-\nabla_1\psi(f)
\end{array}\right)=:\nabla^\perp \psi(x).
\end{equation}

\subsection{$ \mathbf{\delta:\Omega^1\rightarrow\Omega^0} $}\label{ss:delta1to0}
For a vertex $ x\in V_N $
the codifferential results to be 
\begin{equation}\label{eq:cod1-0}
\delta j(x)=\sum_{y:(y,x)\in E_N}j(
y,x).
\end{equation}
To highlight  the analogous nature of this operator to that of \eqref{eq:dpsi}, we define $ \Phi_i(x):=j(x,x+e_i)-j(x-e_i,x) $. Here $ \Phi_i  $ can be thought as the net flow passing trough the "source/sink" $ x  $.
We define a divergence operator as $ \div j(x):=-\delta j(x) $, i.e.
\begin{equation}\label{eq:cod1-0}
\textrm{div\,}j(x)=\us{y:(x,y)\in E_N}{\sum}j(x,y)=\sum_i\Phi_i(x)
\end{equation}
To recover the usual definition of divergence like \eqref{eq:dpsi} with \eqref{eq:ddiv} and a discrete  divergence theorem like \eqref{eq:ddivth3d} we should interpret the divergence in \eqref{eq:cod1-0} with the help of the dual complex. We have to think about (see the operator $ \star $ defined in  footnote  \ref{f:delta}) $ \mathrm{div} j(x) $ as a 3-form on the cube $c^*= *(x) $ and $ j(x,y)  $ as  2-forms on the faces $ f^*=*(x,y) $. With a suitable collection of vertexes $ \{x_i\}_{i=1}^m $ we can define any desired volume $ V^* $.
\subsubsection*{\quad Other dimensions.} In other dimensions definition \eqref{eq:cod1-0} doesn't change.

\bigskip
\bigskip

We conclude determining the  space of harmonic discrete vector fields on the discrete manifold of $ \mathbb{T}^n_N $. 

\begin{proposition}\label{prop:Omega^1_H}
	Consider    on $ \mathbb{T}^n_N $ the collections $ V_N $ and $ E_N^+ $  in $ \mc C $ defined in notation \ref{ss:not}, then in any dimension  d the harmonic space $  \Omega^1_H(\mc {C}) $ has dimension d and it is generated by the discrete vector fields
	\begin{equation}\label{eq:Omega^1_Hbasis}
	\vphi^{(i)}(x,x+e_j):=\delta_{ij},\qquad  i,j=1,\dots, n. 
	\end{equation}
\end{proposition}

\begin{proof}
	For each $ i\in\{1,\dots,n\} $ we have both $ d^1\vphi^{(i)}(f)=0 $ for all $ f\in F_N $ and $ \delta^1 \vphi^{(i)}(x)=0 $ for all $ x\in V_N $, then $ \vphi^{(i)} \in \Omega^1_H $ and the set of discrete vector fields $ \{\vphi^{(i)}(x)\}_i $ generate a \e d-dimensional subspace of $ \Omega_H^1 $.  Since $ \textrm{Ker} \,d^1=d^0\Omega^0 \oplus \Omega^1_H $ we have $ \dim \Omega^1_H=\dim(\textrm{Ker}\, d^1)-\dim(\Ima d^0) $, this is also the dimension of the quotient space $ \Ker d^1/\Ima d^{0}  $  of closed discrete vector fiels (\e 1-forms) differing for  an exact \e 1-form. By duality (see natural pairing \eqref{eq:pairing}) between $ d $ and $ \p $ the quotient space $ \Ker d^1/\Ima d^{0}$ is isomorphic to quotient space of closed \e 1-chains that differ for the  boundary of a \e 2-chain $\Ker \p^1/\Ima \p^{2} $, which dimensionality is  given by the number of independent circle $ S^1 $ in the cartesian product $ S_1\times\dots\times S_1 $ homeomorphic to the continuous torus  $  \mathbb{T}^n $, that is $ n $. Therefore $ \{j_i(x)\}_{i=1}^n $ generates all $ \Omega_H^1 $. \end{proof}

\textbf{Acknowledgements and additional information}

\bigskip

We are grateful to Davide Gabrielli for a very  careful reading of the paper. We also thank Gustavo Granja for
some very useful  discussions.

On behalf of all authors, the corresponding author states that there is no conflict of interest.

\end{document}